\documentclass[a4paper,11pt]{article} %%{proc}
\usepackage{latexsym, amsmath, amssymb}
\usepackage{lscape}
\usepackage{epstopdf,caption,subcaption,graphicx,xcolor}
\usepackage{tikz}
\usepackage{hyperref}
\usepackage{comment}
\usetikzlibrary{shapes,arrows,fit,calc,positioning}
\usetikzlibrary{decorations.pathreplacing}

\normalsize
\newtheorem{theorem}{Theorem}
\newtheorem{lemma}{Lemma}

\newtheorem{definition}{Definition}

\newtheorem{problem}{Problem}
\newenvironment{proof}{{\sc Proof. }}{\hfill$\Box$\vspace{0.1in}}
\newcommand{\mc}[1]{{\cal {#1}}}

\begin{document}

\title{Offline green bin packing and its constrained variant}

\author{Mingyang Gong\footnote{Gianforte School of Computing, Montana State University, Bozeman, MT 59717, USA.
Email: {\tt mingyang.gong@montana.edu}}
\and
Brendan Mumey\footnote{Gianforte School of Computing, Montana State University, Bozeman, MT 59717, USA.
Email: {\tt brendan.mumey@montana.edu}}}

\date{}
\maketitle

%============================================================================== 
\begin{abstract}
In this paper, we study the {\em green bin packing} (GBP) problem where $\beta \ge 0$ and $G \in [0, 1]$ are two given values as part of the input.
The energy consumed by a bin is $\max \{0, \beta (x-G) \}$ where $x$ is the total size of the items packed into the bin.
The GBP aims to pack all items into a set of unit-capacity bins so that 
the number of bins used plus the total energy consumption is minimized.
When $\beta = 0$ or $G = 1$, GBP is reduced to the classic bin packing (BP) problem.
In the {\em constrained green bin packing} (CGBP) problem, the objective is to minimize the number of bins used to pack all items 
while the total energy consumption does not exceed a given upper bound $U$.
We present an APTAS and a $\frac 32$-approximation algorithm for both GBP and CGBP,
where the ratio $\frac 32$ matches the lower bound of BP.
\end{abstract}
%============================================================================

\newpage 

\section{Introduction}\label{sec:intro}
%==========================================================================

In the classic {\em bin packing}~\cite{GJ79} (BP) problem, we are given a set of items where the size of item $i$ is denoted by $s_i \in (0, 1]$
and an infinity number of unit-capacity bins.
The item $i$ can be placed into a bin only when the remaining capacity of the bin is at least $s_i$.
The objective of BP is to minimize the number of bins used to pack all items.
Recently, the authors~\cite{BSS25} proposed a new variant of BP, called {\em green bin packing} (GBP),
in which each used bin incurs a certain amount of energy consumption.
In detail, suppose $G \in [0, 1]$ and the capacity within $[0, G]$ is defined as the {\em green space},
which does not consume any energy.
However, using each unit of the capacity outside the green space will incur energy consumption $\beta \ge 0$.
Therefore, if the total size of the items packed into a bin is $x$, then the energy consumed by the bin is $f_{\beta, G}(x) = \max \{ 0, \beta (x-G) \}$ as shown in Definition~\ref{def01}.
The GBP aims to pack all items into bins so that the number of opened bins plus the total energy consumption is minimized.
One clearly sees that GBP is reduced to BP when $\beta = 0$ or $G = 1$
and thus, GBP must be NP-hard.
As mentioned in~\cite{BSS25}, GBP is motivated by practical scenarios in cloud computing.
One such scenario is {\em greening geographical load balancing}~\cite{LLW11, LLW12}, 
in which each data center has a limited capacity (i.e., the quantity $G$ in GBP) of carbon-free energy
but additional energy is required when the demand exceeds the capacity.

If all sizes as well as $G$ and $\beta$ are given, then both BP and GBP are {\em offline}.
By contrast, in the {\em online} setting, items arrive one by one and each item should be packed into an existing or newly opened bin upon its arrival before the next item comes.
In~\cite{BSS25}, the authors mainly studied the online GBP from the perspective of online algorithms.
Next, we introduce the corresponding results for online BP and GBP even though this paper focuses on the offline setting.

\subsection{Online BP and GBP}

Given an online minimization problem, we let $OPT(I)$ denote the objective value of an optimal offline solution for an instance $I$.
Similarly, we let $A(I)$ denote the objective value of the solution produced by a specific online algorithm $A$ for $I$.
If for every instance $I$, we have $A(I) \le \alpha \cdot OPT(I)$, then we say $A$ is {\em $\alpha$-competitive},
where the value $\alpha$ is also called the {\em absolute competitive ratio}.
On the other hand, if there exists another value $\gamma \ge 0$, independent of the instance $I$, such that $A(I) \le \alpha \cdot OPT(I)+\gamma$, then we say $A$ is {\em asymptotic $\alpha$-competitive} and 
$\alpha$ (respectively, $\gamma$) is the {\em asymptotic competitive ratio} (respectively, {\em additive term}) for $A$.
In other words, $A$ is asymptotic $\alpha$-competitive if $A(I) \le \alpha \cdot OPT(I)$ when $OPT(I)$ tends to infinity.

On one hand, the authors in~\cite{BBG12} showed that no algorithm for online BP can achieve an absolute competitive ratio smaller than $\frac {248}{161} \approx 1.54037$.
On the other hand, online BP has been extensively studied in the literature and many well-known algorithms, for example, {\em Next Fit}, {\em First Fit} and {\em Harmonic}, are presented for it.
Next Fit~\cite{Joh73} places each item into the most recently opened bin if possible; otherwise, it creates a new bin for the item.
The absolute competitive ratio of Next Fit is $2$ and the ratio is tight.
If each item is packed into the earliest opened bin in which it fits, then Next Fit reduces to First Fit~\cite{Joh73}, another fundamental algorithm for BP.
In~\cite{Ull71}, the author showed that First Fit is asymptotic $1.7$-competitive with an additive term $3$.
The ratio $1.7$ is asymptotic tight since there exists an instance $I$ such that $OPT(I) = 10$ and the items in $I$ can be packed into $17$ bins by First Fit~\cite{Ull71}.
The additive term $3$ experienced several improvements.
Johnson et al. improved the additive term from $3$ to $2$~\cite{JDU74}, which was further improved to $0.9$ in~\cite{GGJ76} and to $0.7$ in~\cite{XT10}.
Finally, D\'osa and Sgall~\cite{DS13} proved that for every instance $I$, First Fit uses at most $\lfloor 1.7 \cdot OPT(I) \rfloor$ bins to pack all items in $I$,
which settled the additive term in the competitive analysis of First Fit.
Given a parameter $H$, the Harmonic algorithm first partitions the interval $(0, 1]$ into $H$ sub-intervals $(0, \frac 1H]$, $(\frac 1H, \frac 1{H-1}]$, $\ldots$, $(\frac 13, \frac 12]$ and $(\frac 12, 1]$.
Each item $i$ is placed into an existing bin when the bin only contains items whose sizes belong to the same sub-interval as $i$;
if no such bin exists, then a new bin is created to pack the item $i$.
When $H$ tends to a sufficient large value, the Harmonic algorithm achieves an absolute competitive ratio of $1.69103$~\cite{LL85}.

For online GBP, the authors in~\cite{BSS25} distinguished two cases $\beta G \le 1$ and $\beta G > 1$.
Roughly speaking, when $\beta G \le 1$, the value of $\beta$ is relatively small and 
thus, it is reasonable to use fewer bins while the capacity outside the green space in these bins can be widely used.
In this case, GBP is close to the classic BP and therefore, the algorithms for online BP also perform well for online GBP.
For example, the absolute competitive ratio of Next Fit is $\frac {2+\beta(1-G)}{\beta+(1-\beta G)}$~\cite{BSS25}, 
which equals to $2$ when $\beta = 0$ or $G = 1$, matching the result for online BP.
On the contrary, when $\beta G > 1$, it costs more to occupy non-green space in a bin.
Therefore, we generally prefer to open a new bin, which is different from online BP.
In~\cite{BSS25}, the authors introduced another parameter $\tau \in [0, 1-G]$ and 
each item $i$ with size $s_i \ge G+\tau$ is placed to a single bin separately.
Afterwards, we can reset the capacity of each bin to $G+\tau$ and consequently, the algorithms for online BP can be applied to online GBP too.
For instance, Next Fit achieves an absolute competitive ratio of $\max \{ \frac {G(1+\tau \beta)}{G+\tau}, \frac {G(2+\tau \beta)}{G+2\tau}, \frac {2+\tau \beta}{1+\tau \beta} \}$
and the authors in~\cite{BSS25} also determine the optimal value of $\tau$ that minimizes the ratio.
Interested readers may refer to~\cite{BSS25} for more algorithmic results for online GBP.

\subsection{Offline GBP and its constrained variant}

Given an offline minimization problem, we say that an offline algorithm $A$ is {\em $\alpha$-approximation} if for every instance $I$, we have $A(I) \le \alpha \cdot OPT(I)$
where $A(I)$ and $OPT(I)$ denote the objective values of the solution produced by $A$ and an optimal solution, respectively.
The value $\alpha$ is referred to as the {\em absolute approximation ratio}.
Similarly to online algorithms, if $A(I) \le \alpha \cdot OPT(I)+\gamma$, then we say $A$ is {\em asymptotic $\alpha$-approximation},
where $\alpha$ is called the {\em asymptotic approximation ratio} and $\gamma \ge 0$ is the {\em additive term}.
An {\em asymptotic polynomial time approximation scheme} ({\em APTAS}) is a family of offline approximation algorithms
such that given a fixed $0 < \epsilon \le 1$, it contains an approximation algorithm $A$ that runs in polynomial time in the input size and $A(I) \le (1+\epsilon)OPT(I)+\gamma$ for some $\gamma \ge 0$.
If the additive term $\gamma$ is $0$ in the definition of APTAS, then APTAS becomes a {\em polynomial time approximation scheme} ({\em PTAS}).

For offline BP, no polynomial-time algorithm can have an absolute approximation ratio smaller than $\frac 32$ unless NP=P~\cite{GJ79}.
Simchi-Levi~\cite{Sim94} presented the famous {\em First Fit Decreasing} ({\em FFD)} algorithm
that first sorts the items in the non-increasing order of their sizes and then applies the First Fit algorithm to pack the items.
Luckily, FFD is a $\frac 32$-approximation algorithm, matching the lower bound of offline BP.
The authors in~\cite{FL81} developed an APTAS for offline BP.

In the remainder of this paper, GBP always refers to offline GBP and an algorithm always refers to an offline algorithm for simplicity.
In~\cite{BSS25}, Bibbens et al. assume that both $\beta$ and $G$ are two fixed constants.
Based on this assumption, they presented an APTAS for GBP for every fixed $\beta$ and $G$.
Given a fixed constant $\epsilon>0$, they set $\delta = \min \{ \frac \epsilon2, \frac {\epsilon(1-\beta G)}{2\beta} \}$ when $\beta G \le 1$
and $\delta = \frac \epsilon {4\beta^2}$ otherwise.
It is worth mentioning that the running time is $O(n^{1+s^{1/\delta}})$ where $s=\frac 1{\delta \epsilon} (1+\beta(1-G))$.
We generalize their setting by considering $\beta$ and $G$ as part of the input (See Problem~\ref{prob01} for the detailed definition).
So, the APTAS in~\cite{BSS25} no longer works, since its running time becomes exponential.
We also study a constrained variant of GBP, referred to as CGBP, in Problem~\ref{prob02}.
In CGBP, besides $\beta$ and $G$, we are given an upper bound $U$
and the problem aims to minimize the number of bins used to pack all items such that the total energy consumption is bounded by $U$.
One sees that in the description of CGBP, $\beta$ can be omitted if we set $U' = U/\beta$, $\beta' = 1$ and $G' = G$.
But we keep using $\beta$, $G$ and $U$ for consistency with GBP.

\subsection{Our contributions}

In this paper, given an optimal packing $\sigma^*$ for GBP or CGBP, we present an algorithm that produces a set of polynomially many bin packings in polynomial time
such that there exists a bin packing $\sigma$, in which the number of bins used is at most $1+\epsilon$ times that of $\sigma^*$ plus $1$
and the total energy consumption does not exceed that of $\sigma^*$.
For GBP, the algorithm returns a bin packing from all bin packings produced with the minimum objective value,
while for CGBP, the algorithm returns a bin packing that minimizes the number of bins used subject to the energy constraint.
Recall the properties of $\sigma$ and clearly, our algorithm is an APTAS for both GBP and CGBP.
Recall again that $\beta$, $G$ and $U$ are all part of the input and therefore, our APTAS for GBP extends the one in~\cite{BSS25}.
Using the APTAS for GBP or CGBP, we develop the first $\frac 32$-approximation algorithm, which matches the lower bound for the classic BP problem.

The remainder of this paper is organized as follows.
Section~2 gives some definitions and notations for GBP and CGBP.
Sections~3 and~4 present an APTAS for GBP and CGBP when $G \ge \frac \epsilon3$ and $G < \frac \epsilon3$, respectively where $0 < \epsilon \le 1$ is a fixed constant.
In Section~5, we show an algorithm for GBP and CGBP with an absolute approximation ratio of $\frac 32$.
Lastly, we conclude the paper in Section~6 and summarize some future research problems.

\section{Preliminaries}
\label{sec2}
%==================================================================================================

In the {\em green bin packing} (GBP for short) problem, each bin has a capacity of $1$ and 
is associated with two parameters $G \in [0, 1]$ and $\beta \ge 0$.
There are $n$ items $\mc{I}$ with sizes $s_1, \ldots, s_n$, which should be packed into a set of bins.
Without loss of generality, we can assume $1 \ge s_1 \ge \ldots \ge s_n > 0$.
It is worth mentioning that both $\beta$ and $G$ are assumed to be fixed constants in~\cite{BSS25};
in contrast, our formulation treats both of them as part of the input.
One clearly sees that our formulation generalizes the one in~\cite{BSS25}.
The capacity within $[0, G]$ is referred to as the {\em green space} and 
each unit of capacity beyond the green space incurs an energy consumption of $\beta$.
In other words, if the {\em load}, i.e., the total size of the items packed in the bin, is $x$, then the energy consumed by the bin is $f_{\beta, G}(x) = \max \{ 0, \beta (x-G) \}$ as defined in Definition~\ref{def01}.

\begin{definition}
\label{def01}
The {\em energy consumption} function $f_{\beta, G}: [0, 1] \to \mathbb{R}_{\ge 0}$ maps the load of a bin to a non-negative real number, 
which can be expressed as 
$f_{\beta, G}(x) = \max \{0, \beta (x-G) \}$.
Clearly, given $\beta \ge 0$ and $G \in [0, 1]$, the function $f_{\beta, G}(x)$ is non-decreasing in $x$.
\end{definition}

The GBP aims to minimize the number of bins opened plus the total energy consumption.

\begin{problem}
\label{prob01}
{\em (Green Bin Packing, GBP for short)}
Given a set of items $\mc{I} = \{ 1, \ldots, n \}$ and two parameters $\beta \ge 0$, $G \in [0, 1]$, 
each item $i$ has a size of $s_i \in (0, 1]$ and 
the items should be packed into a set of unit capacity bins with the objective of minimizing 
the number of opened bins plus the total energy consumption (See $f_{\beta, G}$ in Definition~\ref{def01}).
\end{problem}

The {\em constrained green bin packing} (CGBP for short) problem is a variant of GBP,
whose objective is to minimize the number of opened bins 
such that the total energy consumption does not exceed a given upper bound $U \ge 0$.
One clearly sees that if $\beta=0$ or $G = 1$, then $f_{\beta, G}(x)=0$ for any $x \in [0, 1]$,
which indicates that the energy consumption is always $0$.
So, both GBP and CGBP are reduced to the classic bin packing problem.

Recall that $f_{\beta, G}(x)$ is non-decreasing in $x$ and 
therefore, packing each item in a single bin yields a bin packing with the least energy consumption $\sum_{i=1}^n f_{\beta, G}(s_i)$.
It follows that CGBP admits a feasible solution if and only if $U \ge \sum_{i=1}^n f_{\beta, G}(s_i)$.

\begin{problem}
\label{prob02}
{\em (Constrained Green Bin Packing, CGBP for short)}
Given a set of items $\mc{I} = \{ 1, \ldots, n \}$ and three parameters $\beta \ge 0$, $G \in [0, 1]$, $U \ge \sum_{i=1}^n f_{\beta, G}(s_i)$, 
each item $i$ has a size of $s_i \in (0, 1]$ and
the items should be packed into a set of unit capacity bins with the objective of minimizing 
the number of opened bins 
while ensuring that the total energy consumption (See $f_{\beta, G}$ in Definition~\ref{def01}), is bounded by $U$.
\end{problem}

%==================================================================================================

%In this section, we will present an APTAS for GBP (see Problem~\ref{prob01}).
We emphasize that $\beta$, $G$ and $U$ are all regarded as part of the input in GBP and CGBP.
Let $\epsilon \in (0, 1]$ be a fixed constant and for simplicity, we assume $\frac 1\epsilon$ is an integer.
Let $\Delta \le \max \{ G, \frac \epsilon3 \}$ be a function of $\epsilon$ or $G$ and in each later section, we will show the expression of $\Delta$ before we use it.
Assume that $\mc{I}$ is an instance of GBP or CGBP.
Next, we distinguish the {\em large, medium} and {\em tiny} items in $\mc{I}$.

\begin{definition}
\label{def02}
An item $i$ is {\em large} if $s_i > \max \{ G, \frac \epsilon3 \}$.
Similarly, an item $i$ is {\em medium} (respectively, {\em tiny}) if $\Delta < s_i \le \max\{ G, \frac \epsilon3 \}$ (respectively, $s_i \le \Delta$).
We let $\mc{L}$ (respectively, $\mc{M}$ and $\mc{T}$) be the set of large (respectively, medium and tiny) items.
%Given a bin packing $\sigma$, the set of large (respectively, medium and tiny) items packed in the bin $i$ is denoted as $\mc{L}_i(\sigma)$ (respectively, $\mc{M}_i(\sigma)$ and $\mc{T}_i(\sigma)$).
%The {\em large load} (respectively, {\em medium load} and {\em tiny load}) of the bin $i$ in $\sigma$ is the total size of the large (respectively, medium and tiny) items packed in the bin.
\end{definition}

%Given a bin packing $\sigma$ for the instance $\mc{I}$ of GBP or CGBP, let $b(\sigma)$ denote the number of bins used in $\sigma$
%and $E(\sigma)$ be the total energy consumed by $\sigma$.
The bins in a bin packing for $\mc{I}$ can be partitioned into three categories: 
A bin containing at least one large item;
a bin with load less than $G$ and;
a bin with load at least $G$ but containing no large item.
We define these three types of bins below.

\begin{definition}
\label{def03}
A {\em large-item bin} is a bin consisting of at least one large item.
Clearly, the load of a large-item bin is more than $G$.
A {\em light bin} is a bin with load less than $G$.
If a bin is neither a large-item bin nor a light bin, then we say the bin is {\em (large-free) heavy bin}. 
That is, a heavy bin contains no large item but the load is at least $G$.

Given a packing $\sigma$, we let $b_L(\sigma)$ (respectively, $b_h(\sigma)$ and $b_\ell(\sigma)$) denote the number of large-item bins
(respectively, heavy bins and light bins) used in $\sigma$.
Moreover, let $b(\sigma)$ and $E(\sigma)$ denote the total number of bins and the total energy consumption of $\sigma$, respectively.
So, $b(\sigma) = b_L(\sigma)+b_h(\sigma)+b_\ell(\sigma)$.
\end{definition}

Clearly, by Definitions~\ref{def01} and~\ref{def03}, each light bin incurs no energy consumption.
In fact, if the total load of the large-item and heavy bins in $\sigma$ is $K$, then 
the total energy consumption $E(\sigma)$ is $\beta K - \beta G (b_L(\sigma)+b_h(\sigma))$.

For simplicity, for any subset $\mc{I}' \subseteq \mc{I}$ of items, we let $S(\mc{I}')$ denote the total size of the items in $\mc{I}'$.
Moreover, we fix an optimal bin packing $\sigma^*$ for the instance $\mc{I}$ of GBP or CGBP,
which has the minimum energy consumption among all optimal packings.
Given the packing $\sigma^*$, we will present an algorithm that always produces a set of polynomially many bin packings in polynomial time
such that there exists a bin packing $\sigma$
with $b(\sigma) \le (1+\epsilon)b(\sigma^*)+1$ and $E(\sigma) \le E(\sigma^*)$.
For GBP, the algorithm outputs a bin packing with the minimum objective value as defined in Problem~\ref{prob01},
while for CGBP, the algorithm outputs a bin packing that satisfies the energy constraint and uses the minimum number of bins.
Due to the existence of $\sigma$, our algorithm is an APTAS for both GBP and CGBP.
We will show such a $\sigma$ exists in the later sections.
Without loss of generality, we assume $\sigma^*$ has no empty bins and thus $b(\sigma^*) \le n$.
Later, we may transform $\sigma^*$ into different bin packings, in which empty bins are permitted.

\section{An APTAS for GBP and CGBP when $G \ge \frac \epsilon3$}

In this section, we present an APTAS for GBP and CGBP when $G \ge \frac \epsilon3$.
By Definition~\ref{def02}, a large item has a size greater than $G$.

\subsection{Handling large items}

We first assume $|\mc{L}| \ge \frac {24}{\epsilon^2}$.
Recall that $s_1 \ge s_2 \ge \ldots \ge s_n$ and the set $\mc{L}$ of large items consists of the first $|\mc{L}|$ items in $\mc{I}$.
Let $m = \lceil \frac {\epsilon^2 |\mc{L}|}{24} \rceil$.
Note that the large item has a size at least $\max \{ G, \frac \epsilon3 \} \ge \frac \epsilon3$.
Thus, the number of bins in $\sigma^*$ is at least $\frac \epsilon3 |\mc{L}|$.
By $|\mc{L}| \ge \frac {24}{\epsilon^2}$, we have 
\begin{equation}
\label{eq01}
m \le \frac {\epsilon^2 |\mc{L}|}{24} + 1 \le \frac {\epsilon^2}{24} |\mc{L}| + \frac {\epsilon^2}{24} |\mc{L}|
= \frac \epsilon 4 \cdot \frac \epsilon3 |\mc{L}| \le \frac \epsilon 4 \cdot b(\sigma^*).
\end{equation}

We perform a {\em linear grouping} on the large items.
The first group $\mc{L}_1$ contains the first $m$ items in $\mc{L}$ and the second group $\mc{L}_2$ contains the next $m$ items.
We repeatedly do the above grouping until all items have been placed into a group.
Let $\theta$ be the number of groups.
Note that $|\mc{L}_i| = m$ for each $i \le \theta-1$ and $|\mc{L}_\theta| \le m$.
Therefore, we have $m (\theta-1) \le |\mc{L}| \le m \theta$.
It is not hard to see that $3 \le \theta \le \frac {24}{\epsilon^2}+1$.
For each group $\mc{L}_i$, we round the size of each item in $\mc{L}_i$ up to the largest item size in $\mc{L}_i$.
For each $1 \le i \le \theta$, $\mc{L}'_i$ denotes the group corresponding to $\mc{L}_i$ with its rounded sizes and $\mc{L}' = \cup_{i = 1}^\theta \mc{L}'_i$.

Note that the number of groups in $\mc{L}'$ is $\theta$, too.
Clearly, $|\mc{L}'_i| = |\mc{L}_i| = m$ for each $i \le \theta-1$ and $|\mc{L}'_\theta| = |\mc{L}_\theta| \le m$.
For simplicity, we let $\mc{B}_1$ be the set of bins in $\sigma^*$, each of which contains at least one item in $\cup_{i=1}^{\theta-2} \mc{L}_i$.
Moreover, let $\mc{B}_2$ be the set of bins in $\sigma^*$, each of which contains at least one large item from $\mc{L}_{\theta-1}$ and no large items from $\mc{L} \setminus \mc{L}_{\theta-1}$.
So, $\mc{B}_1 \cap \mc{B}_2 = \emptyset$.
For each large-item bin $j \notin \mc{B}_1 \cup \mc{B}_2$, it consists of at least one item in $\mc{L}_\theta$.
Since $|\mc{B}_2| \le |\mc{L}_{\theta-1}| = m$ and $|\mc{L}_\theta| \le m$, we have $\max \{ |\mc{B}_2|, |\mc{L}_\theta| \} \le m$.
We can conclude that 
\begin{equation}
\label{eq02}
b_L(\sigma^*) \le |\mc{B}_1| + |\mc{B}_2| + |\mc{L}_\theta| 
\le  |\mc{B}_1| + \min \{ |\mc{B}_2|, |\mc{L}_\theta| \} + m.
\end{equation}

Next, we show a mapping $\phi$ from the items in $\cup_{i=2}^\theta \mc{L}'_i$ to $\cup_{i=1}^{\theta-1} \mc{L}_i$.
For each $2 \le i \le \theta-1$, $|\mc{L}'_i| = |\mc{L}_{i-1}| = m$.
Thus, the items in $\mc{L}'_i$ can be mapped to $\mc{L}_{i-1}$ one by one, following their order.
For the items in $\mc{L}'_\theta$, we first choose $\min \{ |\mc{L}'_\theta|, |\mc{B}_2| \}$ items in $\mc{L}_{\theta-1}$ assigned to the bins in $\mc{B}_2$,
with at most one item chosen from each bin in $\mc{B}_2$.
By $|\mc{L}'_\theta| \le m = |\mc{L}_{\theta-1}|$, the items in $\mc{L}'_\theta$ can be mapped to $\mc{L}_{\theta-1}$ 
such that the selected items from $\mc{L}_{\theta-1}$ lie in the image of the mapping $\phi$.
One sees that for $2 \le i \le \theta$, each item in $\mc{L}'_i$ is mapped to an item in $\mc{L}_{i-1}$.
Therefore, $s_j \le s_{\phi(j)}$ for each item $j \in \cup_{i=2}^\theta \mc{L}'_i$.

\begin{lemma}
\label{lemma01}
There exists a packing $\sigma_1$ for $\mc{L}' \cup \mc{M} \cup \mc{T}$ such that 
\begin{itemize}
\item[1.] $b_L(\sigma_1) \ge b_L(\sigma^*)$, $b_h(\sigma_1) \ge b_h(\sigma^*)$ and $b(\sigma_1) \le (1+\frac \epsilon4) b(\sigma^*)$;

\item[2.] the total energy consumption of $\sigma_1$ is at most $E(\sigma^*)+\beta (S(\mc{L}') - S(\mc{L}))$.
\end{itemize}
\end{lemma}
\begin{proof}
We construct the packing $\sigma_1$ from the optimal packing $\sigma^*$ for $\mc{I}$:
First, we assign each item in $\mc{M} \cup \mc{T}$ to the same bin as in $\sigma^*$.
We open a new bin to pack each item in $\mc{L}'_1$.
Finally, each item $j$ in $\cup_{i=2}^\theta \mc{L}'_i$ is assigned to the bin where the item $\phi(j)$ is placed in $\sigma^*$.
By $s_j \le s_{\phi(j)}$, $\sigma_1$ is a feasible bin packing for $\mc{L}' \cup \mc{M} \cup \mc{T}$.

We first prove $b_L(\sigma_1) \ge b_L(\sigma^*)$.
Note that by the mapping $\phi$, all the bins in $\mc{B}_1$ and at least $\min \{ |\mc{L}'_\theta|, |\mc{B}_2| \}$ bins in $\mc{B}_2$ are still large-item bins in $\sigma_1$.
Moreover, the newly opened $|\mc{L}'_1|$ bins in $\sigma_1$ are large-item bins.
Recall that $\theta \ge 3$ and thus, $|\mc{L}'_1| = m$.
Therefore, by $|\mc{L}'_1| = m$ and Eq.(\ref{eq02}), we have $b_L(\sigma_1) \ge |\mc{B}_1| + \min \{ |\mc{L}'_\theta|, |\mc{B}_2| \} + m \ge b_L(\sigma^*)$.
Note that each heavy bin in $\sigma^*$ remains heavy in $\sigma_1$.
It follows that $b_h(\sigma_1) \ge b_h(\sigma^*)$.
By Eq.(\ref{eq01}), $b(\sigma_1) \le b(\sigma^*) + |\mc{L}'_1| = b(\sigma^*)+m \le (1+\frac \epsilon4) b(\sigma^*)$.
The first statement is proved.

Next, we prove the second statement.
Let $K$ be the total load of the light bins in $\sigma^*$.
So, the total size of the items assigned to the large-item or heavy bins in $\sigma^*$ is $S(\mc{I})-K = S(\mc{L}) + S(\mc{M} \cup \mc{T}) - K$.
Therefore, the energy consumption $E(\sigma^*)$ equals to $\beta (S(\mc{L})+S(\mc{M} \cup \mc{T})-K) - \beta G (b_L(\sigma^*)+b_h(\sigma^*))$.
By the construction of $\sigma_1$, each light bin in $\sigma^*$ keeps the same elements in $\sigma_1$.
Hence the total load of the light bins in $\sigma_1$ is at least $K$.
By $b_L(\sigma_1) \ge b_L(\sigma^*)$ and $b_h(\sigma_1) \ge b_h(\sigma^*)$, the energy consumption $E(\sigma_1)$ is at most 
\[
\beta (S(\mc{L}')+S(\mc{M} \cup \mc{T})-K) - \beta G (b_L(\sigma^*)+b_h(\sigma^*))
= E(\sigma^*)+\beta (S(\mc{L}') - S(\mc{L})),
\]
which completes the proof.
\end{proof}

When $|\mc{L}| < \frac {24}{\epsilon^2}$, we simply let $\mc{L}' = \mc{L}$.
In this case, the packing $\sigma^*$ for $\mc{L}' \cup \mc{M} \cup \mc{T}$ satisfies the two statements stated in Lemma~\ref{lemma01}.
Moreover, regardless of whether $|\mc{L}| < \frac {24}{\epsilon^2}$, the number of distinct sizes in $\mc{L}'$ is at most $\max \{ \frac {24}{\epsilon^2}, \theta \} \le \frac {24}{\epsilon^2}+1$.

\subsection{Handling medium items}

Similarly to the large items, we handle the medium items in this subsection.
We let $\Delta = \frac {\epsilon^2}{39}$.
By Definition~\ref{def02} and $G \ge \frac \epsilon3$, for each medium item $i$, we have $\Delta < s_i \le G$.
When $|\mc{M}| \ge \frac 8{\epsilon \Delta}$, we perform a linear grouping on $\mc{M}$ too.
We abuse $m = \lceil \frac {\epsilon \Delta |\mc{M}|}8\rceil$
and group every $m$ consecutive medium items in $\mc{M}$ together.
Let $\mc{M}_i$ be the $i$-th group and thus, in each group $\mc{M}_i$ except the last one, the number of items in $\mc{M}_i$ is exactly $m$.
In other words, $\mc{M}$ is partitioned into at most $\frac 8{\epsilon \Delta}+1$ groups.
The group $\mc{M}'_i$ can be obtained from $\mc{M}_i$ by rounding the size of each item in $\mc{M}_i$ up to the largest item size in $\mc{M}_i$.
For simplicity, we let $\mc{M}' = \cup_{i \ge 1} \mc{M}'_i$.
Next, we show a similar result as in Lemma~\ref{lemma01}.

\begin{lemma}
\label{lemma02}
There exists a packing $\sigma_2$ for $\mc{L}' \cup \mc{M}' \cup \mc{T}$ such that 
\begin{itemize}
\item[1.] $b(\sigma_2) \le (1+\frac \epsilon2) b(\sigma^*)$ and;

\item[2.] the total energy consumption of $\sigma_2$ is at most $E(\sigma^*)+\beta (S(\mc{L}') - S(\mc{L}))$.
\end{itemize}
\end{lemma}
\begin{proof}
We construct the packing $\sigma_2$ from the packing $\sigma_1$ in Lemma~\ref{lemma01} for $\mc{L}' \cup \mc{M} \cup \mc{T}$:
Each item in $\mc{L}' \cup \mc{T}$ is packed into the same bin as in $\sigma_1$.
Then, each item in $\mc{M}'_1$ is placed alone in a new bin.
Note that for each $i \ge 2$, $|\mc{M}'_i| \le |\mc{M}_{i-1}|$ and thus, 
we can map the items in $\mc{M}'_i$ to $\mc{M}_{i-1}$ one by one, following their order.
Therefore, the size of an item in $\cup_{i \ge 2} \mc{M}'_i$ does not exceed the size of its mapped item.
We finally pack each item in $\cup_{i \ge 2} \mc{M}'_i$ to the bin where its mapped item is packed in $\sigma_1$.
One sees that $\sigma_2$ is a feasible bin packing for $\mc{L}' \cup \mc{M}' \cup \mc{T}$.

Since $|\mc{M}| \ge \frac 8{\epsilon \Delta}$ and each medium item has a size at least $\Delta$, we have $b(\sigma^*) \ge \Delta |\mc{M}| \ge \frac 8\epsilon$.
It follows that 
\[
m \le \frac {\epsilon \Delta |\mc{M}|}8 + 1 \le \frac \epsilon8 \cdot b(\sigma^*) + \frac \epsilon8 \cdot b(\sigma^*)
= \frac \epsilon 4 \cdot b(\sigma^*).
\]
By Lemma~\ref{lemma01} and the construction of $\sigma_2$, we have 
\[
b(\sigma_2) \le b(\sigma_1)+m \le (1+\frac \epsilon4) b(\sigma^*) + \frac \epsilon 4 \cdot b(\sigma^*) \le (1+\frac \epsilon2) b(\sigma^*).
\]
The first statement is proved.

For each newly opened bin, it only contains a medium item with a size at most $G$.
Therefore, the load is at most $G$ and the bin will not consume any energy.
On the other hand, for each bin in $\{ 1, \ldots, b(\sigma_1) \}$, its load in $\sigma_2$ is at most that in $\sigma_1$.
Recall that the energy consumption is non-decreasing in the load.
We can conclude that $E(\sigma_2) \le E(\sigma_1)$ and the lemma is proved by the second statement of Lemma~\ref{lemma01}.
\end{proof}

If $|\mc{M}| < \frac 8{\epsilon \Delta}$, then we simply let $\mc{M}' = \mc{M}$.
Therefore, regardless of whether $|\mc{M}| < \frac 8{\epsilon \Delta}$, $\mc{M}'$ has at most $\frac 8{\epsilon \Delta}+1$ distinct sizes.
Moreover, Lemma~\ref{lemma02} is always satisfied in which $\sigma_2 = \sigma_1$ if $|\mc{M}| < \frac 8{\epsilon \Delta}$.

\subsection{The complete APTAS}

Before we present our APTAS, we define the {\em bin type} of a bin $j$ that describes the large and medium items as well as the total size, denoted as $z_j$, of the tiny items packed into the bin in $\sigma_2$. 
The rounded total size, denoted as $\hat{z}_j$, of the tiny items is the value obtained by rounding $z_j$ up to the least multiple of $\Delta$.
Therefore, $\hat{z}_j = i \Delta$ for some $0 \le i \le \frac 1\Delta$.
Recall that the number of distinct sizes in $\mc{L}' \cup \mc{M}'$ is at most $\frac {24}{\epsilon^2}+\frac 8{\epsilon \Delta}+2$.
Since the smallest item size in $\mc{L}' \cup \mc{M}'$ is at least $\Delta$, each bin has at most $\frac 1\Delta$ items in $\mc{L}' \cup \mc{M}'$.
These together allow us to describe the bin types for $\sigma_2$.

\begin{definition}
\label{def04}
Let $s'_1, \ldots, s'_d$ be the distinct sizes in $\mc{L}' \cup \mc{M}'$ where $d \le \frac {24}{\epsilon^2}+\frac 8{\epsilon \Delta}+2$.
A {\em bin type} is a $(d+1)$-tuple $(n_1, \ldots, n_d, i)$ where for each $1 \le j \le d$, $n_j$ is the number of items packed into the bin with a size of $s'_j$
and the rounded total size of the tiny items is $i \Delta$.
Clearly, $n_j \le \frac 1\Delta$ and $i \le \frac 1\Delta$. 
Thus, the number of distinct bin types is at most $(\frac 1\Delta+1)^{d+1}$, which is a fixed constant.

A {\em bin configuration} describes the number of bins for each bin type.
By the first statement of Lemma~\ref{lemma02} and $\epsilon \le 1$, we have $b(\sigma_2) \le 2 b(\sigma^*) \le 2n$.
For each bin type, the number of bins is in $\{ 0, \ldots, 2n \}$.
So, the number of distinct bin configurations is bounded by $(2n+1)^{(\frac 1\Delta+1)^{d+1}}$, which is polynomial in $n$.
\end{definition}

For the packing $\sigma_2$ in Lemma~\ref{lemma02}, its corresponding bin configuration can be found by enumerating all bin configurations.
Subsequently, we can easily pack the large and medium items into the bins.
Without loss of generality, we assume the first $b_L(\sigma_2)$ bins in $\sigma_2$ are large-item bins, followed by $b_h(\sigma_2)$ heavy bins
and the last $b_\ell(\sigma_2)$ bins are light.
We determine the assignments of the tiny items by a linear program as in LP(\ref{eq03}).
Let $x_{ij}$ be a binary decision variable for an item $i \in \mc{T}$ and a bin $j \le b(\sigma_2)$
such that $x_{ij} = 1$ if and only if the item $i$ is placed in the bin $j$ in $\sigma_2$.
The first constraint in LP(\ref{eq03}) indicates that each item $i$ must be packed into some bin $j$
while the second constraint indicates that the total size of the tiny items packed into the bin $j$ is at most $\hat{z}_j$.
Therefore, LP(\ref{eq03}) has a feasible solution and we compute a basic solution of LP(\ref{eq03}) in polynomial time~\cite{Len83}.
Let $\mc{T}_j = \{ i \in \mc{T}: x_{ij} = 1 \}$ for each $j = 1, \ldots, b(\sigma_2)$ 
and $\mc{T}_0 = \mc{T} \setminus \cup_{j \ge 1} \mc{T}_j$.
So, $\mc{T}_0$ contains all tiny items with fractional solutions
and each $\mc{T}_j$ contains all items $i$ with $x_{ij}=1$.

\begin{equation}
\label{eq03}
\begin{aligned}
% \min \quad & \sum_{j=1}^{b_L(\sigma_2)} x_{ij} s_i   \\
& \sum_{j=1}^{b(\sigma_2)} x_{ij} = 1, \quad  \forall i \in \mc{T}, \\
 & \sum_{i \in \mc{T}} x_{ij} s_i   \le \hat{z}_j,  \quad \forall \mbox{ bin } j. 
\end{aligned}
\end{equation}

\begin{lemma}
\label{lemma03}
$S(\mc{T}_0) \le b(\sigma_2) \cdot \Delta$ and for each $1 \le j \le b(\sigma_2)$, we can find a subset $\hat{\mc{T}}_j$ of $\mc{T}_j$
such that $S(\hat{\mc{T}}_j) \le 2\Delta$ and $S(\mc{T}_j \setminus \hat{\mc{T}}_j) \le z_j$.
\end{lemma}
\begin{proof}
We first prove $S(\mc{T}_0) \le b(\sigma_2) \cdot \Delta$.
Since each item in $\mc{T}_0$ has a size at most $\Delta$, it suffices to show $|\mc{T}_0| \le b(\sigma_2)$.
Note that LP(\ref{eq03}) has $|\mc{T}|+b(\sigma_2)$ and thus, in a basic solution, there are at most $|\mc{T}|+b(\sigma_2)$ positive values.
For each item in $\mc{T}_0$, it has at most two positive values and the number of positive values is at least 
$\sum_{j \ge 1} |\mc{T}_j| + 2|\mc{T}_0| = |\mc{T}|+|\mc{T}_0|$.
Therefore, $|\mc{T}|+b(\sigma_2) \ge |\mc{T}|+|\mc{T}_0|$ and we have $|\mc{T}_0| \le b(\sigma_2)$.

Consider a bin $j$.
If $S(\mc{T}_j) \le 2\Delta$, then $\hat{\mc{T}}_j = \mc{T}_j$.
Otherwise, $S(\mc{T}_j) > 2\Delta$.
We greedily choose the items in $\mc{T}_j$ until the total size first exceeds $\Delta$.
The chosen items form the set $\hat{\mc{T}}_j$.
Since the size of a tiny item is at most $\Delta$, we know $\Delta \le S(\hat{\mc{T}}_j) \le 2\Delta$.
By the second constraint of LP(\ref{eq03}) and the definition of $\hat{z}_j$, we have $S(\mc{T}_j) \le \hat{z}_j$ and $\hat{z}_j \le z_j + \Delta$.
It follows that
\[
S(\mc{T}_j \setminus \hat{\mc{T}}_j) = S(\mc{T}_j)-S(\hat{\mc{T}}_j) \le \hat{z}_j - \Delta \le z_j.
\]
The lemma is proved.
\end{proof}

The tiny items can be packed as follows: 
For each $1 \le j \le b(\sigma_2)$, we place the items in $\mc{T}_j \setminus \hat{\mc{T}}_j$ to the bin $j$.
Then, we open a new bin and greedily assign an item in $\cup_{j \ge 1} \hat{\mc{T}}_j \cup \mc{T}_0$ to the bin 
if its load is at most $G$ after adding this item.
We repeatedly do this step until all the items in $\cup_{j \ge 1} \hat{\mc{T}}_j \cup \mc{T}_0$ have been packed.
Let $\sigma'$ be the bin packing produced for $\mc{L}' \cup \mc{M}' \cup \mc{T}$.
We replace back the original size of each item in $\mc{L}' \cup \mc{M}'$ in $\sigma'$ and 
we finally obtain a packing $\sigma$ for the original instance $\mc{I}$.

\begin{theorem}
\label{thm01}
The packing $\sigma$ for $\mc{I}$ satisfies that
\begin{itemize}
\item[1.] $b(\sigma) \le (1+\epsilon) b(\sigma^*) + 1$ and;

\item[2.] the total energy consumption of $\sigma$ is at most $E(\sigma^*)$.
\end{itemize}
\end{theorem}
\begin{proof}
Consider the packing $\sigma'$ for $\mc{L}' \cup \mc{M}' \cup \mc{T}$.
To prove the first statement, it suffices to show $b(\sigma') \le (1+\epsilon) b(\sigma^*)+1$.
By Lemma~\ref{lemma03}, the total size of the items in $\cup_{j \ge 1} \hat{\mc{T}}_j \cup \mc{T}_0$ is bounded by $3 \Delta \cdot b(\sigma_2)$.
Note that each newly opened bin in $\sigma'$, except the last one, has load at least $G-\Delta$; otherwise, the bin can contain at least one more tiny item.
By $\epsilon \le 1$, $G \ge \frac \epsilon3$ and $\Delta = \frac {\epsilon^2}{39}$, the number of newly opened bins in $\sigma'$ is bounded by 
\[
\frac {3\Delta \cdot b(\sigma_2)}{G-\Delta} + 1 \le \frac {3\epsilon}{13-\epsilon} \cdot b(\sigma_2) + 1
\le \frac \epsilon4 \cdot b(\sigma_2)+1.
\]
By the first statement of Lemma~\ref{lemma02} and $\epsilon \le 1$ again, we have
\[
b(\sigma') \le (1+\frac \epsilon4)b(\sigma_2)+1 \le (1+\frac \epsilon4)(1+\frac \epsilon2) b(\sigma^*)+1 \le (1+\epsilon) b(\sigma^*)+1.
\]

Next, we prove the second statement.
Note that the newly opened bin in $\sigma'$ has load at most $G$ and thus, the energy consumed by the bin is $0$.
For each bin $1 \le j \le b(\sigma_2)$, the load of $j$ in $\sigma'$ is at most that in $\sigma_2$.
So, the total energy consumption of $\sigma'$ is bounded by that of $\sigma_2$.
By the second statement of Lemma~\ref{lemma02}, it suffices to show $E(\sigma) \le E(\sigma')-\beta (S(\mc{L}') - S(\mc{L}))$.
Let $K$ and $K'$ be the total load of the large-item bins in $\sigma$ and $\sigma'$, respectively.
One sees that an item is large in $\sigma$ if and only if it is large in $\sigma'$.
Therefore, $b_L(\sigma) = b_L(\sigma')$.
%Hence, the set of large-item bins in $\sigma'$ is the same as $\sigma$.
From $\sigma'$ to $\sigma$, the total size of the large items is reduced by $S(\mc{L}') - S(\mc{L})$
and the size of each medium or tiny item does not increase.
It follows that $K \le K' - (S(\mc{L}') - S(\mc{L}))$.
Since the load of a large-item bin is more than $G$, 
we conclude that the energy consumed by the large-item bins in $\sigma$ is 
\[
\beta K - \beta G \cdot b_L(\sigma) \le \beta K' - \beta G \cdot b_L(\sigma') - \beta (S(\mc{L}') - S(\mc{L})),
\]
where $\beta K' - \beta G \cdot b_L(\sigma')$ is the energy consumed by the large-item bins in $\sigma'$.
Clearly, for each heavy and light bin $j$ in $\sigma$, the load of $j$ is at most that in $\sigma'$.
These together show that $E(\sigma) \le E(\sigma')-\beta (S(\mc{L}') - S(\mc{L}))$.
The theorem is proved.
\end{proof}

As described at the end of Section~2, we have presented an APTAS for GBP and CGBP when $G \ge \frac \epsilon3$.
%By the above theorem, the objective value of $\sigma$ is at most $(1+\epsilon)b(\sigma^*)+E(\sigma^*)+1$ if we consider the GBP
%and thus, we have designed an APTAS for GBP when $G \ge \frac \epsilon3$.
%Note that our APTAS is also an APTAS for CGBP because $E(\sigma) \le E(\sigma^*) \le U$ where $U$ is a given upper bound on the energy consumption.

\section{An APTAS for GBP and CGBP when $G < \frac \epsilon3$}

In this section, we always let $\Delta = G$ and thus, by Definition~\ref{def02}, each large or medium item has a size more than $G$.
We distinguish the following three cases: (1) In $\sigma^*$, there is a light bin. 
In other words, $b_\ell(\sigma^*) \ne 0$;
(2) In $\sigma^*$, there is no light bin and the total size of the tiny items is relatively large.
In detail, $b_\ell(\sigma^*) = 0$ and $S(\mc{T}) \ge 2G (b_h(\sigma_1) - |\mc{M}|)$;
(3) $b_\ell(\sigma^*) = 0$ and $S(\mc{T}) < 2G (b_h(\sigma_1) - |\mc{M}|)$.

\subsection{$b_\ell(\sigma^*) \ne 0$}

In this subsection, we assume $b_\ell(\sigma^*) \ne 0$.
In other words, there exists a light bin in $\sigma^*$.
Recall that the packing $\sigma^*$ achieves the minimum energy consumption among all optimal packings.
The following lemma shows that each large and medium item should be packed alone into a bin in $\sigma^*$.

\begin{lemma}
\label{lemma04}
Each large and medium item must be packed alone into a bin in $\sigma^*$.
\end{lemma}
\begin{proof}
Let $j$ be a bin that contains a large or medium item $i$ in $\sigma^*$.
Assume the lemma does not hold and there exists another item, say $k \ne i$, placed in the bin $j$ in $\sigma^*$.
By $s_i > G$, we know $s_k < 1-G$ and moving the item $k$ to a light bin forms a new bin packing.
Since the load of a light bin is less than $G$ and $s_i > G$ again, removing the item $k$ from the bin $j$ will reduce the energy by $\beta s_k$
and adding it to a light bin will incur an energy strictly less than $\beta s_k$.
Therefore, the new packing has the same number of bins as in $\sigma^*$ but the energy consumption is strictly less than $E(\sigma^*)$,
which contradicts the choice of $\sigma^*$.
The lemma is proved.
\end{proof}

By Lemma~\ref{lemma04}, we can easily place a large and medium item into a bin as in $\sigma^*$.
In the sequel, we remove all the large and medium items from the instance $\mc{I}$ and further remove their corresponding bins in $\sigma^*$.
In other words, we assume each item in $\mc{I}$ is tiny.

\begin{lemma}
\label{lemma05}
The load of each bin in $\sigma^*$ is at most $2G$.
\end{lemma}
\begin{proof}
Assume the contrary and there exists a bin $j$ such that the load is more than $2G$ in $\sigma^*$.
Let $i$ be a tiny item placed in the bin $j$ in $\sigma^*$.
Note that $s_i \le G$.
After removing the item $i$ from $j$, the load of $j$ is more than $G$.
Similarly to Lemma~\ref{lemma04}, moving the item $i$ from $j$ to a light bin will reduce the energy consumption, 
contradicting the choice of $\sigma^*$.
The proof is completed.
\end{proof}

By the above lemma, we know that the load of each bin in $\sigma^*$ is at most $2G$.
Therefore, we can construct a scaled instance $\mc{I}'$ from $\mc{I}$: 
For each item $i$, we set $s'_i = \frac {s_i}{2G}$.
Moreover, let $G' = \frac 12$, $\beta' = 2G \beta$ and $U' = U-\beta \sum_{i \in \mc{L} \cup \mc{M}} (s_i-G)$ for CGBP.
One sees from Definition~\ref{def01} that if $x \le 2G$, then $f_{\beta, G}(x) = f_{\beta', G'}(x')$ where $x' = \frac x{2G}$.
Thus, it suffices to consider the scaled instance $\mc{I}'$.
Note that by $\epsilon \le 1$, $G' = \frac 12 \ge \frac \epsilon3$.
So, the APTAS presented in Section~3 is applicable for $\mc{I}'$.

Now, it suffices to assume that there is no light bin in $\sigma^*$, i.e., $b_\ell(\sigma^*) = 0$.
Therefore, the load of each bin in $\sigma^*$ is at least $G$.
It follows that the energy consumed by $\sigma^*$ is $\beta (S(\mc{I}) - G \cdot b(\sigma^*))$,
where $b(\sigma^*)$ is the number of bins used in $\sigma^*$ and $S(\mc{I})$ is the total size of the items in $\mc{I}$.
That is, 
\begin{equation}
\label{eq04}
E(\sigma^*) = \beta (S(\mc{I}) - G \cdot b(\sigma^*)).
\end{equation}
By $G < \frac \epsilon3$ and Definition~\ref{def02}, an item $i$ is large if $s_i > \frac \epsilon3$.
As described in Section~3.1, we can construct the set of large items $\mc{L}'$,
in which the number distinct sizes is at most $\frac {24}{\epsilon^2}+1$.
Moreover, there exists a packing $\sigma_1$ for $\mc{L}' \cup \mc{M} \cup \mc{T}$ satisfying the two statements in Lemma~\ref{lemma01}.
Similarly to Definition~\ref{def04}, we define the large-item bin types in $\sigma_1$
that describe the set of large items in the large-item bins.
Subsequently, we define the large-item bin configurations, which specify the number of bins for each large-item bin type.

\begin{definition}
\label{def05}
Let $s'_1, \ldots, s'_d$ be the distinct sizes in $\mc{L}'$ where $d \le \frac {24}{\epsilon^2}+1$.
A {\em large-item bin type} is a $d$-tuple $(n_1, \ldots, n_d)$ where each entry $n_j$ indicates the number of items packed into the bin with a size of $s'_j$ in $\sigma_1$.
Since the size of a large item is greater than $\frac \epsilon3$, we have $n_j \in \{ 0, \ldots, \frac 3\epsilon \}$.
Thus, the number of distinct large-item bin types is at most $(\frac 3\epsilon+1)^{d+1}$, which is a fixed constant.

By the first statement of Lemma~\ref{lemma01} and $\epsilon \le 1$, we have $b_L(\sigma_1) \le b(\sigma_1) \le 2 b(\sigma^*) \le 2n$.
A {\em large-item bin configuration} describes, for each large-item bin type, the number of bins used in $\sigma_1$,
where the number is in $\{ 0, \ldots, 2n \}$.
Therefore, the number of distinct large-item bin configurations is at most $(2n+1)^{(\frac 3\epsilon+1)^{d+1}}$, which is polynomial in $n$.
\end{definition}

One clearly sees that we can enumerate all large-item bin configurations and find the one corresponding to $\sigma_1$.
Therefore, we can easily determine the assignments of all large items in $\mc{L}'$.
After replacing back the original size of each large item, we have designed a packing for large items in $\mc{L}$.
Obviously, we still have $b_L(\sigma_1)$ large-item bins.
Since $b_h(\sigma_1) \le 2b(\sigma^*) \le 2n$, the value of $b_h(\sigma_1)$ can be guessed in $O(n)$ time.
In the next two subsections, we assume $b_h(\sigma_1)$ is known to us.

\subsection{$b_\ell(\sigma^*) = 0$ and $S(\mc{T}) \ge 2G (b_h(\sigma_1) - |\mc{M}|)$}

Note that $b_h(\sigma_1)$ is known and we have completed the assignments of the large items.
Moreover, each medium (respectively, tiny) item has a size more than (respectively, at most) $G$.
We first open $b_h(\sigma_1)$ new bins.
Then, we arbitrarily choose $\min \{ b_h(\sigma_1), |\mc{M}| \}$ medium items 
and place each of them alone into a bin.
If $b_h(\sigma_1) > |\mc{M}|$, then for each of the remaining $b_h(\sigma_1)-|\mc{M}|$ bins, 
we greedily fit a tiny item into a bin until the load first exceeds $G$.
So, by $G < \frac \epsilon3 \le \frac 13$, the load is at most $2G < \frac 23$.
By $S(\mc{T}) \ge 2G (b_h(\sigma_1) - |\mc{M}|)$, we can guarantee that 
the load of each newly opened $b_h(\sigma_1)$ bins is more than $G$.
Currently, we have $b_L(\sigma_1)+b_h(\sigma_1)$ bins and each of them has load at least $G$.
If any medium or tiny items remain unassigned, then we fit them into existing bins when possible;
otherwise, a new bin is created until all items are packed.
This way, we obtain a packing $\sigma$ for the input instance $\mc{I}$.

\begin{theorem}
\label{thm02}
The packing $\sigma$ for $\mc{I}$ satisfies that
\begin{itemize}
\item[1.] $b(\sigma) \le (1+\epsilon) b(\sigma^*) + 1$ and;

\item[2.] the total energy consumption of $\sigma$ is at most $E(\sigma^*)$.
\end{itemize}
\end{theorem}
\begin{proof}
We first prove $b(\sigma) \le (1+\epsilon) b(\sigma^*) + 1$.
If $b(\sigma) \le b_L(\sigma_1) + b_h(\sigma_1)$, then by the first statement of Lemma~\ref{lemma01}, 
$b(\sigma) \le b(\sigma_1) \le (1+\frac \epsilon4) b(\sigma^*)$ and we are done.
Next, we assume $b(\sigma) \ge b_L(\sigma_1) + b_h(\sigma_1) + 1$.
By Definition~\ref{def02} and $G < \frac \epsilon3$, the size of a medium or tiny item is at most $\frac \epsilon3$.
So, the last bin is created because the load of each existing bin is at least $1-\frac \epsilon3$; 
otherwise, the bin can add at least one more medium or tiny item.
It follows from $\epsilon \le 1$ and $S(\mc{I}) \le b(\sigma^*)$ that 
\[
b(\sigma) \le \frac {S(\mc{I})}{1-\frac \epsilon3} + 1 \le (1+\epsilon) S(\mc{I}) + 1
\le (1+\epsilon) b(\sigma^*) + 1.
\]

Note that the total number of large-item and heavy bins is at least $b_L(\sigma_1)+b_h(\sigma_1)$.
Therefore, the energy consumed by $\sigma$ is at most $\beta (S(\mc{I}) - G (b_L(\sigma_1)+b_h(\sigma_1)))$.
By the first statement of Lemma~\ref{lemma01}, Eq.(\ref{eq04}) and $b_\ell(\sigma^*) = 0$, we have
\[
\beta (S(\mc{I}) - G (b_L(\sigma_1)+b_h(\sigma_1))) \le \beta (S(\mc{I}) - G (b_L(\sigma^*)+b_h(\sigma^*))) 
= \beta (S(\mc{I}) - G \cdot b(\sigma^*)) = E(\sigma^*).
\]
The proof is finished.
\end{proof}

\subsection{$b_\ell(\sigma^*) = 0$ and $S(\mc{T}) < 2G (b_h(\sigma_1) - |\mc{M}|)$}

The final case is $b_\ell(\sigma^*) = 0$ and $S(\mc{T}) < 2G (b_h(\sigma_1) - |\mc{M}|)$.
We pack each medium item alone into a new bin.
So far, we have determined the assignments for all large and medium items.
Next, we consider the tiny items in $\mc{T}$.

By Definition~\ref{def03}, an item packed into a heavy bin in $\sigma_1$ is either medium or tiny.
Since $S(\mc{T}) < 2G (b_h(\sigma_1) - |\mc{M}|)$, we know $b_h(\sigma_1) > |\mc{M}|$ and
there exist at least $b_h(\sigma_1)-|\mc{M}|$ heavy bins in $\sigma_1$, each of which only contains tiny items.
Let $\mc{B}$ be a collection of $b_h(\sigma_1)-|\mc{M}|$ such heavy bins.
So, $|\mc{B}| = b_h(\sigma_1)-|\mc{M}|$.
For each bin $j \in \mc{B}$, its load is at least $G$ and we can greedily choose the tiny items in $j$ until the total size first reaches or exceeds $G$.
We abuse the notation $\mc{T}_j \subseteq \mc{T}$ to denote the chosen tiny items from $j$.
Since a tiny item has a size at most $G$, we have $S(\mc{T}_j) \le [G, 2G)$.

\begin{lemma}
\label{lemma06}
There exists a packing $\sigma_3$ for $\mc{T}$ such that 
\begin{itemize}
\item[1.] $b(\sigma_3) = b_h(\sigma_1) - |\mc{M}|$ and the load of each bin in $\sigma_3$ is at most $3G$;

\item[2.] $E(\sigma_3) = \beta (S(\mc{T}) - G (b_h(\sigma_1) - |\mc{M}|))$.
\end{itemize}
\end{lemma}
\begin{proof}
We construct the packing $\sigma_3$ as follows:
We open $b_h(\sigma_1) - |\mc{M}|$ bins.
By $|\mc{B}| = b_h(\sigma_1)-|\mc{M}|$, we can establish a bijective mapping $\phi$ from the bins in $\mc{B}$ to our opened bins.
Subsequently, we fit the set $\mc{T}_j$ for each $j \in \mc{B}$ to the bin $\phi(j)$.
So, the load of each bin is $[G, 2G)$.
For the unassigned items in $\mc{T} \setminus \cup_{j \in \mc{B}} \mc{T}_j$, we fit an item into the first bin until the load first exceeds $2G$
and then, repeatedly do this operation for all the bins.
Since a tiny item has a size at most $G$ and $G < \frac \epsilon3 \le \frac 13$, the load is less than $3G < 1$.
By $S(\mc{T}) < 2G (b_h(\sigma_1) - |\mc{M}|)$, we can pack all unassigned items in $\mc{T} \setminus \cup_{j \in \mc{B}} \mc{T}_j$ into the opened bins.

By the construction, the first statement is clearly true.
The second statement holds because each bin in $\sigma_3$ is heavy (its load is at least $G$).
\end{proof}

We perform a similar scaling operation as in Section~4.1, for each item tiny $i \in \mc{T}$, 
the scaled size is $s'_i = \frac {s_i}{3G}$.
Then, we let $G' = \frac 13$, $\beta' = 3G \beta$ and $\epsilon' = \frac \epsilon2$.
For the CGBP problem, the upper bound $U'$ on the energy consumption is set to 
$\beta (S(\mc{T}) - G (b_h(\sigma_1) - |\mc{M}|))$.
By $G' \ge \frac 13 \ge \frac {\epsilon'}3$, the APTAS presented in Section~3 is applicable for the tiny items with scaled sizes,
which outputs a bin packing, denoted as $\sigma_4$ for $\mc{T}$ after replacing back the original sizes.
By $\epsilon' = \frac \epsilon2$, we can guarantee that $b(\sigma_4) \le (1+\frac \epsilon2) b(\sigma_3)+1$ and $E(\sigma_4) \le E(\sigma_3)$.
Combining the packings for $\mc{L}, \mc{M}$ and $\mc{T}$ yields a packing $\sigma$ for the input instance $\mc{I}$.

\begin{theorem}
\label{thm03}
The packing $\sigma$ satisfies that 
\begin{itemize}
\item[1.] $b(\sigma) \le (1+\epsilon) b(\sigma^*) + 1$ and;

\item[2.] the total energy consumption of $\sigma$ is at most $E(\sigma^*)$.
\end{itemize}
\end{theorem}
\begin{proof}
Recall that $b(\sigma_4) \le (1+\frac \epsilon2) b(\sigma_3)+1$ and $E(\sigma_4) \le E(\sigma_3)$.
So, the total number of bins used in $\sigma$ is at most $b_L(\sigma_1)+|\mc{M}|+(1+\frac \epsilon2)b(\sigma_3)+1$.
By $\epsilon \le 1$, Lemmas~\ref{lemma01} and~\ref{lemma06}, we have 
\begin{eqnarray*}
b_L(\sigma_1)+|\mc{M}|+(1+\frac \epsilon2)b(\sigma_3)+1 & \le & (1+\frac \epsilon2)(b_L(\sigma_1)+b_h(\sigma_1))+1 \\
& \le & (1+\frac \epsilon2) b(\sigma_1) + 1 \\
& \le & (1+\frac \epsilon2)(1+\frac \epsilon4) b(\sigma^*)+1 \\
& \le & (1+\epsilon) b(\sigma^*) + 1.
\end{eqnarray*}
The first statement is proved.
Clearly, the energy consumed by large-item bins (respectively, heavy bins) is $\beta (S(\mc{L})-G \cdot b_L(\sigma_1))$
(respectively, $\beta (S(\mc{M})-G \cdot |\mc{M}|)$).
By $E(\sigma_4) \le E(\sigma_3)$ and the second statement of Lemma~\ref{lemma06}, the total energy consumption 
is at most $\beta (S(\mc{I}) - G(b_L(\sigma_1) + b_h(\sigma_1))$.
By the first statement of Lemma~\ref{lemma01}, $b_\ell(\sigma^*) = 0$ and Eq.(\ref{eq04}), we have 
\[
\beta (S(\mc{I}) - G(b_L(\sigma_1) + b_h(\sigma_1)) \le \beta (S(\mc{I}) - G(b_L(\sigma^*) + b_h(\sigma^*))
= \beta (S(\mc{I}) - G \cdot b(\sigma^*)) = E(\sigma^*).
\]
The lemma is proved.
\end{proof}

Combining all three subsections, we have presented an APTAS for GBP and CGBP when $G < \frac \epsilon3$.

\section{A $\frac 32$-approximation for GBP and CGBP}

In this section, we present an algorithm for GBP and CGBP, whose absolute approximation ratio is $\frac 32$.
We still let $\sigma^*$ be an optimal packing that achieves the minimum energy consumption among all optimal packings.
It suffices to compute a bin packing $\sigma$ such that $b(\sigma) \le \frac 32 b(\sigma^*)$ and $E(\sigma) \le E(\sigma^*)$.
When $b(\sigma^*) = 1$, the packing $\sigma^*$ can be obtained by placing the items into a single bin.
If $b(\sigma^*) \ge 3$, then we call our APTAS using $\epsilon=\frac 16$ as input, which outputs a packing $\sigma$
with $b(\sigma) \le \frac 76 b(\sigma^*) + 1$ and $E(\sigma) \le E(\sigma^*)$.
It follows from $\frac 13 b(\sigma^*) \ge 1$ that $b(\sigma) \le \frac 32 b(\sigma^*)$ and we are done.
Subsequently, we can assume $b(\sigma^*) = 2$ and clearly, $S(\mc{I}) \le 2$.

Different from Definition~\ref{def02}, we say an item $i$ is {\em large} if $s_i \ge \frac 13$; otherwise, it is {\em tiny}.
Similarly, different from Definition~\ref{def03}, we say a bin is {\em heavy} if its load is at least $G$; otherwise, it is {\em light}.
We let $\mc{L}$ and $\mc{T}$ denote the set of large and tiny items, respectively.
Recall that $b(\sigma^*) = 2$.
For $j \in \{ 1, 2 \}$, we further let $\mc{L}_j$ (respectively, $\mc{T}_j$) denote the set of large (respectively, tiny) items packed into the bin $j$.
By $S(\mc{I}) \le 2$, $|\mc{L}| \le 6$. 
Thus, the sets $\mc{L}_1$ and $\mc{L}_2$ can be determined in $2^6 = O(1)$ time.
If $|\mc{T}_2| \in \{ 0, 1 \}$, then $\mc{T}_2$ can be found in $O(n)$ time and so is $\sigma^*$.
Therefore, we can assume $|\mc{T}_2| \ge 2$.

The sets $\mc{L}_1$ and $\mc{L}_2$ are known.
Without loss of generality, we assume $S(\mc{L}_1) \ge S(\mc{L}_2)$.
If $S(\mc{T}_1) \ge S(\mc{T}_2)$, then $\mc{T}_1$ and $\mc{T}_2$ can be swapped, after which the load of each bin does not increase.
Since the function $f_{\beta, G}(x)$ in Definition~\ref{def01} is non-decreasing in $x$, the energy consumption does not increase.
Now, it suffices to assume $|\mc{T}_2| \ge 2$, $S(\mc{L}_1) \ge S(\mc{L}_2)$ and $S(\mc{T}_1) \le S(\mc{T}_2)$.

The first case is $S(\mc{T}) \ge 4G$.
By $S(\mc{T}_1) \le S(\mc{T}_2)$, we have $S(\mc{T}_2) \ge 2 G$ and 
clearly, the second bin is heavy in $\sigma^*$.
We claim that the first bin is heavy in $\sigma^*$ too.
Assume the contrary.
By $|\mc{T}_2| \ge 2$, we can let $i$ denote an item in $\mc{T}_2$ with the {\em minimum} size.
So, $s_i \le \min \{ \frac 13, \frac 12 S(\mc{T}_2) \}$.
Since $4G \le S(\mc{T}) \le S(\mc{I}) \le 2$, we have $G \le \frac 12$.
After moving the item $i$ to the first bin, the load of the first bin is at most $G+s_i \le \frac 12+\frac 13 < 1$.
Thus, this moving operation yields a new packing $\sigma_1$ for $\mc{I}$.
By $s_i \le \frac 12 S(\mc{T}_2)$ and $S(\mc{T}_2) \ge 2G$, we have $S(\mc{T}_2 \setminus \{ i \}) \ge G$.
Similarly to Lemma~\ref{lemma04}, removing the item $i$ from the second bin reduces the energy consumption by $\beta s_i$,
while adding it to the first bin increases the energy consumption by strictly less than $\beta s_i$.
Therefore, the energy consumed by $\sigma_1$ is strictly less than that by $\sigma^*$, contradicting the choice of $\sigma^*$.
In conclusion, the two bins in $\sigma^*$ are both heavy.
That is, for each $j \in \{ 1, 2 \}$, $S(\mc{L}_j \cup \mc{T}_j) \ge G$ and $E(\sigma^*) = \beta (S(\mc{I})-2G)$.

Next, we construct a set $\mc{T}'_1 \subseteq \mc{T}$ of tiny items in polynomial time
such that for each $j \in \{ 1, 2 \}$, $S(\mc{L}_j \cup \mc{T}'_j) \ge G$, where $\mc{T}'_2 = \mc{T} \setminus \mc{T}'_1$.
If there exists an item $i \in \mc{T}_1$ with $s_i \ge G$, then the item can be found in $O(n)$ time and we let $\mc{T}'_1 = \{ i \}$.
Since $\mc{T}_2 \subseteq \mc{T}'_2$ and $S(\mc{L}_2 \cup \mc{T}_2) \ge G$, we have $S(\mc{L}_2 \cup \mc{T}'_2) \ge G$.
Therefore, $\mc{T}'_1 = \{ i \}$ is as desired.
If there is no item $i \in \mc{T}_1$ with $s_i \ge G$, then since $S(\mc{L}_1 \cup \mc{T}_1) \ge G$,
$\mc{T}'_1$ can be constructed by greedily choosing the items in $\mc{T}$ until the total size first exceeds $\max \{ 0, G-S(\mc{L}_1) \}$.
One sees that $S(\mc{T}'_1) \le \max \{ G, 2G-S(\mc{L}_1) \} \le 2G$.
By $S(\mc{T}) \ge 4G$ and $\mc{T}'_2 = \mc{T} \setminus \mc{T}'_1$, we have $S(\mc{T}'_2) \ge 2G$.
So, the set $\mc{T}'_1$ is as desired too.
After we compute the sets $\mc{T}'_1$ and $\mc{T}'_2$, for each $j= \{ 1, 2 \}$, by $S(\mc{L}_j \cup \mc{T}'_j) \ge G$, 
we can repeatedly add the items in $\mc{T}'_j$ to the bin $j$ until its load first reaches or exceeds $G$.
Note that each item in $\mc{T}$ is tiny and thus, the size is less than $\frac 13$.
Therefore, by $G \le \frac 12$, the load of each bin is at least $G$ and at most $\max \{ S(\mc{L}_1), S(\mc{L}_2), G+\frac 13 \} \le 1$.
Finally, if there is some unassigned item, then we pack it into an existing bin when possible; otherwise, a new bin is opened until all items are packed.
Let $\sigma$ be the final packing for $\mc{I}$.
Clearly, $b(\sigma) \ge 2$ and the first two bins are heavy by the construction.
In the last paragraph, we prove $E(\sigma^*) = \beta (S(\mc{I})-2G)$.
Therefore, $E(\sigma) \le \beta (S(\mc{I})-2G) = E(\sigma^*)$.
If the third bin is opened, then each of the first two bins has load at least $\frac 23$.
By $S(\mc{I}) \le 2$, the third bin has load bounded by $\frac 23$ and thus, three bins are enough.
It follows that $b(\sigma) \le 3 = \frac 32 b(\sigma^*)$.
Until now, we have designed a $\frac 32$-approximation algorithm when $S(\mc{T}) \ge 4G$.

The remaining case is $S(\mc{T}) < 4G$.
We let $\hat{\mc{T}} = \{ i \in \mc{T}: s_i \le \frac 14 G \}$.
For each item $i \in \mc{T} \setminus \hat{\mc{T}}$, $s_i > \frac 14 G$ and thus, $|\mc{T} \setminus \hat{\mc{T}}| < 16$.
So, each item in $\mc{T} \setminus \hat{\mc{T}}$ can be assigned to the same bin as in $\sigma^*$ in $2^{16} = O(1)$ time.
For each $j \in \{ 1, 2 \}$, let $\hat{\mc{T}}_j$ be the set of items in $\hat{\mc{T}}$ assigned to the bin $j$ in $\sigma^*$
and let $k_j$ denote the maximum integer such that $S(\hat{\mc{T}}_j) \ge \frac {k_j}4 G$.
Clearly, $\frac {k_j}4 G \le S(\hat{\mc{T}}_j) < \frac {k_j + 1}4 G$.
By $S(\hat{\mc{T}}_j) \le S(\mc{T}) < 4G$, we have $k_j \in \{ 0, \ldots, 16 \}$.
So, the values of $k_1$ and $k_2$ can be determined in constant time.
We greedily fit the items of $\hat{\mc{T}}$ into the bin $j$ as long as their total size is no more than $\frac {k_j}4 G$.
Since $s_i \le \frac 14 G$ for each $i \in \hat{\mc{T}}$, the total size of the items in $\hat{\mc{T}}$ packed into the bin $j$ is in $(\frac {k_j-1}4 G, \frac {k_j}4 G]$.
By $S(\hat{\mc{T}}_j) < \frac {k_j + 1}4 G$, the items in $\hat{\mc{T}}$ that remain unassigned have a total size at most $G$.
So, we only need one more bin to pack those items without consuming any energy.
We still use $\sigma$ to denote the packing produced.
Obviously, $b(\sigma) \le 3 = \frac 32 b(\sigma^*)$ and by $S(\hat{\mc{T}}_j) \ge \frac {k_j}4 G$, for each of the first two bins in $\sigma$, its load does not exceed that in $\sigma^*$.
Since the function $f_{\beta, G}(x)$ in Definition~\ref{def01} is non-decreasing in $x$, we have $E(\sigma) \le E(\sigma^*)$.
In summary, we have designed a polynomial-time $\frac 32$-approximation algorithm for GBP and CGBP,
which is best possible since even the classic bin packing cannot be approximated within a factor of $\frac 32$ unless NP=P.

\section{Conclusions}
%==================================================================================================

In this paper, we study the GBP and CGBP problem where the values of $\beta$, $G$ and $U$ are all part of the input.
Our main contributions are an APTAS and a $\frac 32$-approximation algorithm for these two problems.

In~\cite{BSS25}, the authors proved that the absolute approximation ratio for GBP is at least $\frac {1.5+\beta (1-G)}{1+\beta (1-G)}$.
Note that our $\frac 32$-approximation algorithm achieves this lower bound when $\beta = 0$ or $G = 1$.
But it is still interesting to prove whether or not the lower bound is tight for some other values of $\beta$ and $G$.
Moreover, it would also be interesting to extend our model to non-uniform bin capacities
where the function $f_{\beta, G}$ in Definition~\ref{def01} may be defined in terms of the portion of occupied capacity in a bin.

\bibliography{main.bib}

\end{document}